\documentclass[letterpaper, 11pt]{article}
\usepackage[margin=1in]{geometry}
\usepackage{amsmath}
\usepackage{hyperref}
\usepackage{amsthm}
\usepackage{adjustbox}
\usepackage[margin=0pt]{caption}
\usepackage{changepage}
\usepackage{setspace}
\usepackage{tikz}
\usepackage{xcolor}
\usepackage{cancel}
\usetikzlibrary{arrows.meta}
\usepackage{bbding}
\usepackage{enumitem}
\usepackage{amstext}
\usepackage{tikz}
\usetikzlibrary{shapes,snakes}
\usepackage{amsmath}
\newcommand{\footremember}[2]{%
    \footnote{#2}
    \newcounter{#1}
    \setcounter{#1}{\value{footnote}}%
}
\usepackage{array}
\usepackage{tabu}
\usepackage[T1]{fontenc}
\usepackage{babel}
\usepackage{array}
\usepackage{amssymb}
\usepackage{url}
\usepackage{amsfonts}
\usepackage{xspace}

\usepackage{amsthm}

\usepackage{algorithm}
\usepackage[noend]{algorithmic}
\usepackage[colorinlistoftodos,prependcaption,textsize=tiny]{todonotes}

\usepackage{thm-restate}

\newtheorem{theorem}{Theorem}
 \newtheorem{lemma}[theorem]{Lemma}

 \newtheorem{definition}[theorem]{Definition}

\def\T{\mathcal{T}}
\def\I{\mathcal{I}}

\def\P{\mathcal{P}}

\DeclareMathOperator{\OPT}{OPT}

\DeclareMathOperator{\EXP}{\mathbb{E}}

\date{}

\newcommand{\old}[1]{{{}}}

\newcommand{\riko}[1]{{\color{blue}Riko: #1}}
\renewcommand{\riko}[1]{}

\newcommand{\todoin}[1]{}

\renewcommand{\todo}[1]{}

\date{} 

\title{An Algorithm for Bichromatic Sorting with Polylog Competitive Ratio}

\author{Mayank Goswami\footremember{mayank}{Queens College CUNY, Flushing, New York, USA. Supported by NSF grant CCF-1910873. \texttt{mayank.goswami@qc.cuny.edu}}
\and Riko Jacob\footremember{riko}{IT University of Copenhagen, København S, Denmark. Part of this work done during the second Hawaiian workshop on parallel algorithms and data structures, University of Hawaii at Manoa, Hawaii, USA, NSF Grant CCF-1930579. \texttt{rikj@itu.dk}}}

\begin{document}
\maketitle
\begin{abstract}
The problem of sorting with priced information was introduced by [Charikar, Fagin, Guruswami, Kleinberg, Raghavan, Sahai (CFGKRS), STOC 2000]. In this setting, different comparisons have different (potentially infinite) costs. The goal is to find a sorting algorithm with small competitive ratio, defined as the (worst-case) ratio of the algorithm's cost to the cost of the cheapest proof of the sorted order.

The simple case of bichromatic sorting posed by [CFGKRS] remains open:
We are given two sets $A$ and $B$ of total size $N$, and the cost of an $A-A$ comparison or a $B-B$ comparison is higher than an $A-B$ comparison. The goal is to sort $A \cup B$. An $\Omega(\log N)$ lower bound on competitive ratio follows from unit-cost sorting. Note that this is a generalization of the famous nuts and bolts problem, where $A-A$ and $B-B$ comparisons have infinite cost, and elements of $A$ and $B$ are guaranteed to alternate in the final sorted order.

In this paper we give a randomized algorithm InversionSort with an almost-optimal w.h.p. competitive ratio of $O(\log^{3} N)$. This is the first algorithm for bichromatic sorting with a $o(N)$ competitive ratio. 
\end{abstract}

\section{Introduction and main results}

In their seminal paper ``Query strategies for priced information'', Charikar, Fagin, Guruswami, Kleinberg, Raghavan and Sahai \cite{charikar2002query}  [CFGKRS STOC 2000] study the problem of computing a function $f$ of $n$ inputs, where querying an input has a certain cost associated to it, and one wants to find the cheapest query strategy that computes $f$. The \textit{competitive ratio} is defined as the (worst case) ratio of the cost of the query strategy to the cost of the cheapest proof of $f$. This work initiated a multitude of papers on priced information, studying problems like learning with attribute costs \cite{kaplan2005learning}, stochastic boolean function evaluation \cite{deshpande2014approximation}, searching on trees \cite{onak2006generalization,mozes2008finding}, priced information in external memory \cite{bender2021batched}, and others. 

The problem of \textit{sorting with priced information}, or even a simple bichromatic version of it (stated by (CFGKRS \cite{charikar2002query})), remains tantalizingly open. After describing their main results, (CFGKRS \cite{charikar2002query}) mention in further directions: \textit{"Sorting items when each comparison has a distinct cost appears to be highly non-trivial. Suppose, for example, we construct an instance of this problem by partitioning the items into sets A and B, giving each A-to-B comparison a very low cost, and giving each A-to-A and B-to-B comparison a very high cost. We then obtain a very simple non-uniform cost structure in the spirit of the well-known hard problem of ‘‘sorting nuts and bolts’’"}

We will call the above problem \textit{bichromatic sorting}. Observe that bichromatic sorting is a special case of the  \textbf{general} \textit{sorting with priced information} problem defined as follows. Given a weighted undirected complete graph $G$, with weights $w_{ij} \in \mathbb{R}^{\geq 0} \cup \{\infty\}$ indicating the cost to compare keys (vertices) $x_i$ and $x_j$, find the cheapest sorting algorithm. 
To define the competitive ratio of an algorithm, one first needs to define the cheapest proof, which here is simply the sum of the costs to compare pairs of keys that are adjacent in the final sorted order. Restated, it is the (finite) cost of the directed Hamiltonian path that must be present (if the keys can be sorted) in the underlying DAG $\vec G$ resulting after revealing the directions of all edges in $G$. Note that the same weighted complete graph $G$ can give rise to different DAGs $\vec G$ depending on the values assigned to the keys; we write ${\vec G} \leftarrow G$ if there is a valid assignment of values such that $\vec G$ is obtained from $G$ after revealing all comparisons. The competitive ratio of an algorithm $A$ on a given family $\mathcal{G}$ of weighted complete graphs is then defined as
\[\max_{G \in \mathcal{G}} \max_{{\vec G} \leftarrow G} \frac{\text{Cost}_{G}(A,\vec G)}{\text{Cost}_{G}(H_{\vec G})} ,\]

\noindent where $\text{Cost}_{G}(A,\vec G)$ is the cost of the algorithm $A$ on the input with underlying DAG $\vec G$, and $\text{Cost}_{G}(H_{\vec G})$ is the cost of the Hamiltonian path in $\vec G$.

With this notation, bichromatic sorting corresponds to the family $\mathcal{G}$ consisting of all graphs  $G = (V=R \cup B, E= E_{rr} \cup E_{rb} \cup E_{bb})$, where the vertex set can be partitioned into R (Red) and B (Blue), and the edges are of three types: $E_{rr}$ is the set of red-red edges with cost $\alpha \geq 0$, $E_{rb}$ the red-blue edges with cost $1$, and $E_{bb}$ the blue-blue edges with cost $\beta \geq 0$ (w.l.o.g. we have normalized the bichromatic comparisons to cost $1$). The classic \textbf{unit-cost} setting corresponds to the complete graph with all weights equal to one. The competitive ratio in the unit-cost setting is thus $\Theta(\log n)$\footnote{This is because the Hamiltonian Path costs $n-1$, and comparison-based sorting costs $\Theta(n \log n)$.}. 

This paper addresses the question posed by (CFGKRS \cite{charikar2002query}): \textit{What is the competitive ratio for bichromatic sorting. Is it close to the $\Theta(\log n)$ competitive ratio in the unit-cost setting?}

\subsection{Related Work} Prior to this work, we do not know of any algorithm for bichromatic sorting with a $o(n)$ competitive ratio. Apart from bichromatic sorting, another special case of sorting with priced information was termed generalized sorting by Huang, Kannan, and Khanna~\cite{6108244}. Here all costs are in $\{1,\infty\}$; the input is an undirected graph $G$ on the $n$ vertices representing the keys to be sorted, and the $m$ edges of $G$ indicate the allowed comparisons (with the $\infty$ cost edges missing in $G$, the ``forbidden'' comparisons).
The goal is to sort the input while querying as few edges as possible. 
\cite{6108244} gave the first subquadratic algorithm for this $\{1,\infty\}$ cost setting that has a total query cost of ${\widetilde O}(n^{1.5})$, or equivalently, has a competitive ratio of ${\widetilde O}(\sqrt{n})$. This ratio was recently improved by Kuszmaul and Narayanan \cite{kuszmaulnarayanan} to $O(\sqrt{m/n})$.

Gupta and Kumar~\cite{gupta2001sorting} consider the setting where the cost to compare two keys is well-behaved; specifically, the cost function is monotone in the weight of the two keys being compared, examples being the sum and the product. \cite{6108244} also studied the stochastic version of generalized sorting, when each edge in $G$ exists with probability $p$. Subsequent work considered settings where $G$ is dense \cite{banerjee_et_al:LIPIcs:2016:6044} and the setting where the costs induce a metric space over the elements \cite{purohit2018improving}.

\noindent\textit{Nuts-and-Bolts:} In this problem one is given $n$ nuts (R) and $n$ bolts (B), is only allowed to compare a nut to a bolt (R-B), and is \emph{promised a matching} between the nuts and the bolts. The goal is to find this matching. Note that the result of a comparison can be $<, >$ or $=$. The problem is originally mentioned as an exercise in \cite{rawlins1992compared}, page 293, and a simple Quicksort type algorithm can be shown to solve this problem in $O(n \log n)$ comparisons with high probability: Pick a random nut, compare to all bolts, find the matching bolt, and compare that bolt to all nuts. The problem is now partitioned into two subproblems with the match at the boundary; recurse. In a later work by Alon, Blum, Fiat, Kannan, Naor, Ostrovsky 
\cite{alon1994matching} the authors developed a sophisticated deterministic algorithm in $O(n \text{ polylog } n)$ time, which was then improved to an optimal $O(n \log n)$ by Koml\'{o}s, Ma, Sz\'{e}meredi \cite{komlos1998matching}. Algorithms for nuts-and-bolts problem cannot be used for bichromatic sorting because of the lack of matches.
 
A preliminary version of this paper \cite{goswami2023sorting}, that includes some further results, is published on arxiv.org.
 
\subsection{Our Results}

Let us consider bichromatic sorting with two sets $R$ (Red) of size $n$ and $B$ (Blue) of size $m$, and normalize the cost of $R-B$ comparisons to be $1$. Let $N=n+m$ denote the total size of the two sets.
Let $\alpha$ and $\beta$ denote the cost of $R-R$ and $B-B$ comparisons. Observe that if $\alpha<1<\beta$ or $\beta<1<\alpha$, then this can be solved as a monotonic structured cost~\cite{gupta2001sorting}. The case when $\alpha<1$ and $\beta<1$ is discussed in Appendix~\ref{sec:other_bichromatic}. Thus, the $\alpha>1$ and $\beta>1$ version posed in \cite{charikar2002query}, where bichromatic comparisons are the cheapest, is the most interesting variant of bichromatic sorting, and our focus here.
Surprisingly, no results are known for it, even though this is just ``one step up" from unit-cost sorting! We show
 
\begin{restatable}[Polylog Competitiveness of InversionSort]{theorem}{thmbichromaticinversionsort}\label{thm:bichromaticinversionsort}
    There exists an algorithm InversionSort for bichromatic sorting and a constant $c>0$, such that for every instance~$\I$, the cost of InversionSort on~$\I$, with probability at least $1-1/N$, is at most $c(\log N)^3 H_\I$, where $H_\I$ is the cost of the Hamiltonian.
\end{restatable}

Clearly, the $\Omega(\log N)$ lower bound on the competitive ratio carries over from unit-cost sorting. The above result shows that bichromatic sorting is almost as easy as the unit-cost setting. We remark that unlike most algorithms for generalized sorting (costs in $\{1,\infty\}$) that have a high polynomial runtime, InversionSort runs in $O(N^2)$ time. Further, we sketch (Appendix~\ref{sec:other_bichromatic}) how InversionSort can be extended to the multichromatic-cost setting too. This result opens up the search for other, more complex variants that also admit polylog competitive ratios.

\noindent\textbf{Future Directions:} We believe that the $O(\log^{3} N)$ competitive ratio of InversionSort can likely be reduced to an optimal $O(\log N)$. While InversionSort is randomized, it would also be very interesting to construct a deterministic algorithm for bichromatic sorting with similar guarantees. As we mention, bichromatic sorting generalizes the nuts-and-bolts problem, for which considerable effort was required \cite{alon1994matching,komlos1998matching} to match the simple randomized $O(N \log N)$ run time by a deterministic algorithm. We believe that due to the absence of matches, there are interesting obstacles in this direction.

\noindent\textbf{Roadmap:} In Section 2 we describe InversionSort. Section 3 sketches the proof of the guarantee presented in Theorem~\ref{thm:bichromaticinversionsort}. Finally, Section 4 gives the detailed analysis. Due to space constraints some proofs are moved to the appendix.
\section{InversionSort: Definition} \label{sec:bichromatic}

We first start by describing the setup of InversionSort. 
The design principle behind InversionSort is to focus on identifying the monochromatic stretches of the Hamiltonian called \emph{stripes}.
Intuitively, InversionSort operates like a variant of Quicksort, where the control flow is BFS-like instead of the usual DFS-like recursive control flow.
InversionSort focuses on using bichromatic comparisons as much as possible.
A generic state of InversionSort will be defined using a \emph{backbone}, which is a sequence of compared elements of alternating colors, called \emph{representatives} or \emph{pivots}. 
Each representative will be assigned a \emph{bucket}, which is a set of elements of the same color, that lie between the two neighboring representatives of the other color on the backbone.
This backbone is refined if new inversions are found, giving the algorithm its name.
This leads to a ternary \emph{refinement tree}, where every node stands for a subproblem that was split in three by finding a new inversion.
In this first (and only interesting) phase of the algorithm, monochromatic comparisons are only used to increase the probability of finding an inversion or to verify that no further inversions exist.
Here, we use a classical technique of balancing costs that does for each red-red comparison $\alpha$ bichromatic comparisons, and for each blue-blue comparison $\beta$ bichromatic comparisons

We first develop some notation to better explain InversionSort and explain how the algorithm works. 
We summarize this discussion in pseudocode as Algorithm~\ref{alg:InversionSort}.
For simplicity of exposition, we add an artificial smallest red element and an artificial largest blue element to the input.
Doing comparisons with these artificial elements does not incur any cost.

InversionSort makes progress from one state to the next by performing three steps: a) finding an \emph{inversion} (which we define soon)
between neighboring representatives on the backbone, b) inserting this inversion on the backbone, and c) pivoting with these two elements, thereby refining the buckets. 

\begin{definition}\label{def:backbone}
[Backbone, Representatives, and Buckets]
The backbone consists of a totally ordered, alternating list of \emph{representatives}
$(u_0,u_1,u_2,\ldots,u_{2k},u_{2k+1})=(r_0,b_1,r_2,\ldots,r_{2k},b_{2k+1})$,
where $r_{2i} \in R$ and $b_{2j+1} \in B$ 
with $r_{i}<b_{i+1}$ and $b_i <r_{i+1}$.
Here, $r_0$ is an artificial red element that is smaller than all elements, and the last element $b_{2k+1}$ is an artificial blue element that is larger than all elements.
The representatives define the \emph{buckets}
$(X_0,X_1,X_2,\ldots,X_{2k},X_{2k+1})=(R_0,B_1,R_2,\ldots,R_{2k},B_{2k+1})$ by 
$R_i = \{ x\in R \mid b_{i-1} < x < b_{i+1} \} \setminus \{r_{i}\}$
and 
$B_i = \{ x\in B \mid r_{i-1} < x < r_{i+1} \} \setminus \{b_{i}\}$.
Here, as a convention, the representative is not included in the bucket. Again, 
$R_0 = \{ x\in R \mid x < b_{1} \}$
and
$B_{2k+1} = \{ x\in B \mid r_{2k} < x \}$
are special cases.
\end{definition}

\begin{definition}\label{def:active}
    [active subproblems and buckets]
    As long InversionSort did not create a certificate that there are no further inversions between $x_i$ and $x_{i+1}$, the \emph{subproblem} defined by the buckets $X_i$ and $X_{i+1}$ is called active, and so are $X_i$ and $X_{i+1}$.
\end{definition}

From now on, we will use $u_i$ and $X_i$ when we don't care about the color of the $i$th representative and bucket, otherwise we will switch to $r_{i}/b_{i}$ and $R_{i}/B_{i}$. We observe that the buckets of two non-adjacent representatives on the backbone are ordered in the same way as the representatives:
\begin{lemma}\label{lem:backboneTrans}
    Let $X_i$ and $X_j$ be buckets (of arbitrary color) and assume  $i+1\le j-1$.
    Then for all $x\in X_i$ and $y\in X_j$ we have $x<y$.
\end{lemma}
\begin{proof}
    By definition, $x<u_{i+1}$, and $u_{j-1}<y$.
    If $i+1 < j-1$, then $u_{i+1}< u_{j-1}$ is implied by transitivity on the backbone, otherwise $u_{i+1}$ is the same element as $u_{j-1}$.
    In either case, we get that $x < u_{i+1} \le u_{j-1} < y$. 
\end{proof}

\subsection{Inversions and how to find them}\label{sec:inversionSearch}

We now define an \emph{inversion}, which gives the algorithm its name. 
Consider adjacent representatives $u_i$ and $u_{i+1}$, their corresponding adjacent buckets~$X_i$ and $X_{i+1}$, and a bichromatic pair $(x,y)$ of elements $x\in X_i$ and $y\in X_{i+1}$.
Unlike in Lemma~\ref{lem:backboneTrans}, $x$ and $y$ are not ordered by transitivity of the backbone.
Because $x$ and $y$ are of different color, they can be compared at cost 1, which is cheaper than any monochromatic comparison.
If $y<x$, we call the pair an \emph{inversion}. This is because  it allows us to extend the backbone: we get
$u_i<y<x<u_{i+1}$, which is a chain of actual comparisons between elements of alternating color. %

The  simplest way to find an inversion is to uniformly at random, from all pairs in $X_i$ and $X_{i+1}$, pick $x$ and $y$.
If the fraction of inversions is $p$, then the probability of finding an inversion is~$p$ and the expected number of trials to find one is~$1/p$.

If $\alpha$ and~$\beta$ are really large, it is efficient to use only bichromatic comparisons to find an inversion.
Otherwise, some monochromatic comparisons can help.
Using a well known balancing technique, we proceed in rounds and do one (or constantly many) bichromatic comparison per round, and every~$\alpha$ rounds  one red-red comparison, every~$\beta$ rounds one blue-blue.
Then, in an amortized sense, the monochromatic comparisons contribute cost one per round.
These monochromatic comparisons are used to maintain the max / min of a sample of elements in the bucket as follows.
The representative of the bucket is always considered sampled, and hence it is initially the max and min of the sample.
More precisely, we maintain the sample $S_i\subset X_i$ with $u_i\in S_i$.
When sampling uniformly $x\in X_i \setminus S_i$, we perform one or two monochromatic comparisons to check if it is the new max or min of~$S_i$.
If it is, say $x=\max S_i$, we (bichromatically) test $y=\min S_{i+1} < x$.
If this test is successful, we found an inversion $u_i<y<x<u_{i+1}$.
To get the randomness of this inversion into our framework, we use only one of its endpoints, say~$y$ to split~$X_i$, by comparing all elements of~$X_i$ with~$y$ leading to $X_i'<y<X_i'' (<u_{i+1})$.
Because $x\in X_i''$ and $u_i\in X_i'$, both are nonempty.
Now we choose uniformly at random~$x'\in X_i''$ and use $y,x'$ as the inversion.

In every round of comparisons, we also chose a random $x \in X_i$ and compare $(u_{i-1} <) x < \max S_{i-1}$ and $\min S_{i+1} < x (< u_{i+1})$.
If one of them is true, we found a sufficiently random inversion.

If there is no further inversion to be found, this random sampling must fail.
For two buckets with $A$~red and $B$~blue elements, w.l.o.g. red smaller than blue, there are four potential proofs that no further inversion exists:
\begin{enumerate}
    \item all $AB$ red-blue comparisons at total cost~$AB$.
    \item $A-1$ red-red comparisons to find the maximal red, and $B$ red-blue comparisons with this maximum, total cost $\alpha(A-1) + B$.
    \item $B-1$ blue-blue comparisons to find the minimal blue, and $A$ red-blue comparisons with this minimum, total cost $\beta(B-1) + A$.
    \item $A-1$ red-red comparisons to find the maximal red, $B-1$ blue-blue comparisons to find the minimal blue, and 1 red-blue comparison between the maximum and the minimum, total cost~$\alpha(A-1) + \beta(B-1) + 1$. (This case is never the cheapest if $\alpha,\beta>1$)
\end{enumerate}
Observe that the cost of these proofs only depends on the sizes of the stripes/buckets.
Hence, InversionSort computes these values and chooses the cheapest.
If the accumulated cost of randomized inversion finding exceeds the cost of this cheapest proof, it (deterministically) performs all comparisons required by the proof.
If this indeed shows that there are no further inversions, this subproblem is finished, i.e. no longer active.
Else, this attempt to establish a proof must find one or more inversions.
Similar to the case of finding an inversion $\min S_{i+1} < \max S_{i}$, We choose an inversion that is living up to our randomness requirements (as spelled out in Lemma~\ref{lem:uniform}).

Observe that the Hamiltonian is a proof of type~4 and that hence the chosen proof is at most as expensive as that part of the Hamiltonian.
This fits to the purpose of the proof to only identify the stripes, not (yet) to sort the instance.

\subsection{Description of InversionSort}
\begin{algorithm}
    \caption{Algorithm InversionSort}
    \label{alg:InversionSort}
    \begin{algorithmic}
        \REQUIRE elements $R$ red, $B$ blue
        \STATE create trivial backbone $\mathcal{B}$ from $R$ and $B$, see Definition~\ref{def:backbone}
        \STATE $r \leftarrow 0$
        \WHILE{there is an active subproblem in $\mathcal{B}$}
        \STATE $r \leftarrow r+1$
        \IF{$r \mod \alpha == 0$ }
            \FOR{each active red bucket (see Definition~\ref{def:active})}
                \STATE Sample one more red element, update max/min using red-red comparisons
            \ENDFOR
        \ENDIF
        \IF{$r \mod \beta == 0$ }
            \FOR{each active blue bucket}
                \STATE Sample one more red element, update max/min using blue-blue comparisons
            \ENDFOR
        \ENDIF
        \FOR{each active (see Definition~\ref{def:active}) bucket $s$}
            \STATE Sample one element~$x_s$
        \ENDFOR
        \FOR{each active subproblem between buckets $s$ (left), and $q$ (right)}
            \STATE Test for inversion between, use the first one of
            \STATE ~
            \begin{minipage}{9cm}
                \begin{enumerate}[noitemsep]
                    \item $x_s$ and $x_q$
                    \item $x_s$ and $\min\{ \hbox{sampled in~}q \}$
                    \item $\max\{ \hbox{sampled in~}q \}$ and $x_q$
                    \item $\max\{ \hbox{sampled in~}q \}$ and $\min\{ \hbox{sampled in~}q \}$
                \end{enumerate}                    
            \end{minipage}
        \ENDFOR
        \FOR{each active subproblem where $r$ - mark (age) > cheapest certificate based on sizes (Section~\ref{sec:inversionSearch})}
        \STATE do the comparisons of the cheapest certificate
        \IF{this leads to a certicate}
        \STATE the subproblem is finished, i.e. no longer active
        \ELSE
        \STATE at least one inversion is found
        \ENDIF
        \ENDFOR 
        \FOR{each found inversion}
        \STATE update the backbone, including splitting buckets and resampling pivots (Section~\ref{sec:inversionSearch})
        \STATE mark new subproblems with round $r$, 
        \ENDFOR
        \ENDWHILE
        \STATE sort the monochromatic stripes individually
    \end{algorithmic}
\end{algorithm}

InversionSort starts by (trivially) having the backbone consist only of the artificial smallest red element~$r_0$ and largest blue element~$b_1$, and $R_0=R$ and $B_1=B$.

For a given backbone $(u_0,u_1,u_2,\ldots,u_{2k},u_{2k+1})=(r_0,b_1,r_2,\ldots,r_{2k},b_{2k+1})$, InversionSort first, for each pair $X_{i},X_{i+1}$ of adjacent buckets that have no
proof that there is no further inversion, in a round-robin manner, does one round of comparison with amortized constant cost, as spelled out in Section~\ref{sec:inversionSearch}.
If this leads to an inversion, the inverted pair is saved and the algorithm moves to the next pair of adjacent buckets.

At the end of the round, all identified inversions are considered and used to extend the backbone.
Then InversionSort splits existing buckets by pivoting with new elements on the backbone. 
Because there is at most one pair of inversions between each two neighboring representatives on the backbone, each element is compared to at most two new representatives in each round.

This reestablishes the backbone and creates some new pairs of neighboring buckets, for which we initialize the inversion finding procedures.
Existing buckets (identified by the representative) might get smaller. %
Here, a new inversion finding might reuse some of the already established samples and their structure. 
In the analysis we present here, this is not used, so for the sake of simplicity, we assume the algorithm starts from scratch.

The algorithm stops once all neighboring pairs of buckets are no longer active, i.e., shown to not have an inversion.
In this case, the current buckets form the stripes of the instance. 
The algorithm finishes by sorting these buckets (including the representatives) using an optimal number of additional monochromatic comparisons.

\section{InversionSort Analysis: Proof Structure}

 The analysis of InversionSort constitutes the technically challenging part of the paper. 
 Let us visualize a run of InversionSort as a ternary (refinement) tree, where nodes correspond to subproblems. 
 For an internal node~$v$, there is a corresponding subinterval on the backbone defined by two consecutive pivots, say a blue pivot followed by a red pivot, $b_{v} < r_{v}$.  If InversionSort finds an inversion~$y<x$ ($x$ is blue and $y$ is red) between~$b_v$ and $r_v$, then~$v$ has three children with the respective pivots $(b_v,y)$, $(y,x)$, $(x,r_v)$. Note that at a snapshot somewhere during a run of InversionSort, the tree explored so far may be far from being a complete tree - InversionSort could be stuck on a large-sized problem in one region, while refining and working way down a descendant of an interval in another region. Hence InversionSort does not proceed layer-by-layer on this tree.

 The random nature of the inversion searching of InversionSort, as made precise in Lemma~\ref{lem:uniform} leads to the following insight:

\begin{restatable}[Height of the refinement tree]{theorem}{thmheight}
    \label{thm:height}
    Let $\T$ be the refinement tree of running InversionSort on an instance~$\I$ with $N=n+m$ elements.
    With high probability in~$N$, the height of~$\T$ is $O(\log N )$.
\end{restatable}

As a second challenge in the analysis, because of the overlapping nature of the problem, InversionSort cannot easily focus on elements between neighboring representatives.
For example, for the child indicated by pivots $(y,x)$, instead of only getting the reds and blues that actually lie in this range as input, InversionSort instead has to also work with the red elements contained in $(b_v,y)$ and the blue elements inside $(x,r_v)$. This ``spill-in'' from the neighboring subintervals on the backbone needs to be analyzed. See Figure~\ref{fig:stripes1} for an example.
 Thus we distinguish between \emph{subinstances} and \emph{subproblems}. 

\begin{figure}[h]
    \begin{center}
    \includegraphics[width=14cm]{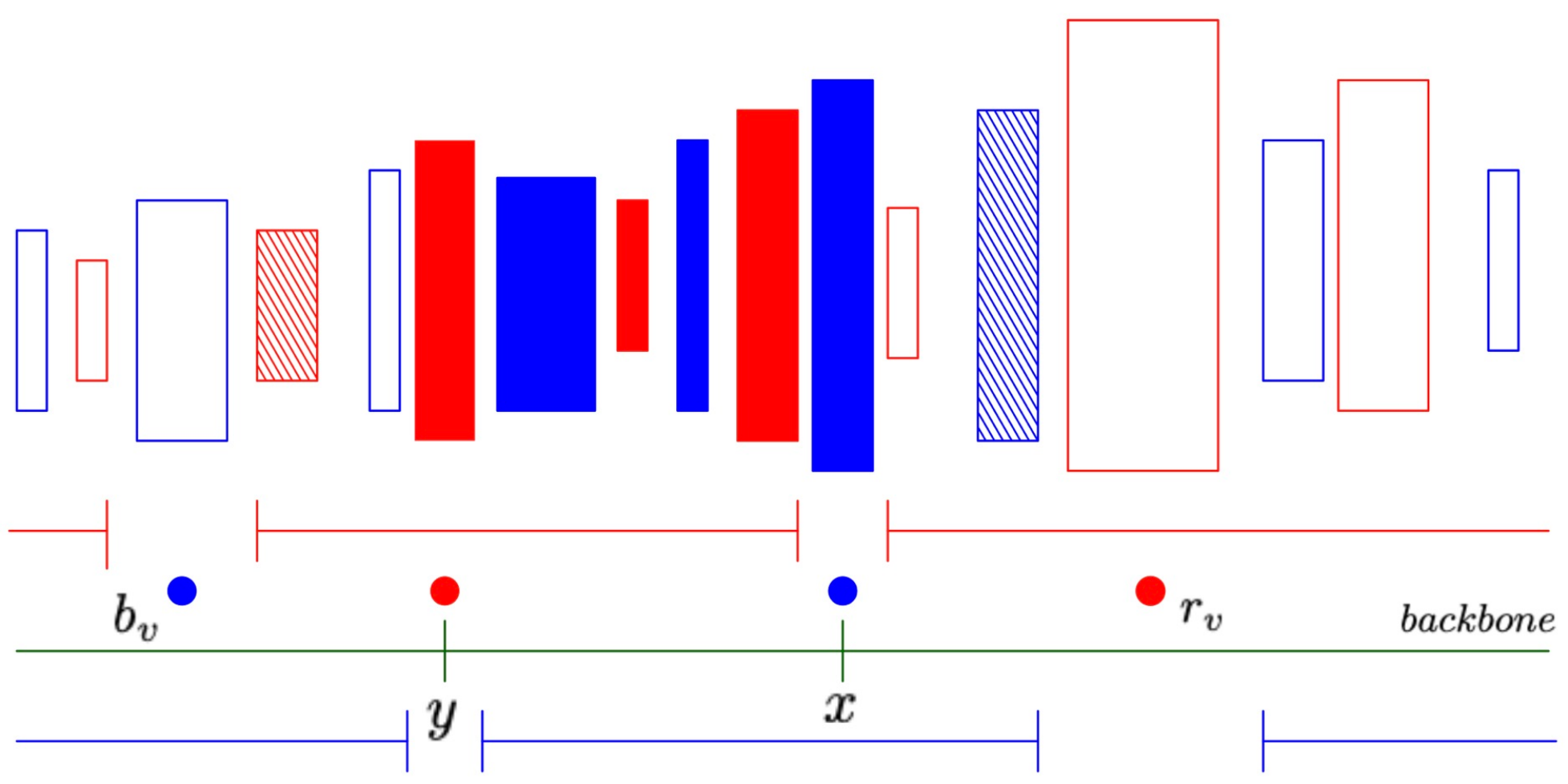}
    \end{center}
    \caption{The difference between a subproblem and a subinstance. $b_v$ and $r_{v}$ are on the backbone, and $y<x$ is the inversion found between them. The stripes in the subinstance $(y,x)$ are solid, whereas the stripes in the subproblem $(y,x)$ are the union of the solid and shaded stripes.}
    \label{fig:stripes1}
\end{figure}

As we will argue, the cost of inversion search procedure of InversionSort is justified by the subinstance between the neighboring elements on the backbone. 
However, if the spill-in for this subproblem is too large, inversion search is too costly.
Hence, we will identify subproblems that do not have too much spill-in from their neighbors - we call these subproblems \emph{unaffected}. 
Inversion search in unaffected subproblems can be charged to their subinstance. What remains is to then show that many subproblems are actually unaffected. In Lemma~\ref{lem:noLongWO} we show that at any time, with high probability, at least roughly a $1/(\log N)^2$ fraction of all current problems are unaffected. This requires a careful accounting of how the unaffected nodes are distributed in the tree. 

\paragraph{Putting everything together} As mentioned above, a $(\log N)^2$ factor appears while accounting (with high probability) for the affected nodes on one snapshot of the backbone. 
Accounting over the whole tree, including the pivoting, introduces another $\log N$ factor corresponding to the depth of the tree. 
\section{InversionSort: Detailed Analysis}
\label{sec:inv_opt_bichromatic}

\begin{definition}[Stripe]\label{def:stripe}
    Let $x$ be an element of an instance $\I$ of bichromatic sorting.
    The the \emph{stripe} of $x$ consists of all elements~$y$ of~$\I$ where no element of the other color is between $x$ and~$y$.
\end{definition}

\noindent\textbf{Refinement tree as a trace and mean to analyze:}
Looking only at the backbone, the algorithm does an interval refinement. 
If we look back in time, the final backbone is given by the instance, only the choice of the representative of a stripe is arbitrary (everything else is determined by the instance).
In contrast, the evolution of the backbone is very much driven by the random choices of the algorithm and leads to a hierarchy of intervals on the backbone. 
A new inversion and the pivoting steps split one interval of the backbone (not to be confused with a bucket that somehow spans two intervals of the backbone) into three, as detailed in the following definition.

\begin{definition}
    [Refinement tree]
    The root node of the tree has as bounding pivots the artificial smallest red and largest blue element.
    Every node~$v$ of the tree has two pivots~$b_v$ and $r_v$ of the (final) backbone associated to it with the guarantee that $b_v$ and $r_v$ were neighbors on the backbone at some stage of the algorithm.
    If $r_v<b_v$ the node is said to have \emph{polarity} ``red smaller than blue'', otherwise ``blue smaller than red''.
    If the algorithm finds an inversion~$x,y$ between~$b_v$ and $r_v$, then~$v$ has three children with the respective pivots $(b_v,y)$, $(y,x)$, $(x,r_v)$ or $(r_v,y)$, $(y,x)$, $(x,b_v)$.
    The polarity of the middle child is the opposite of the polarity of~$v$, whereas the two outer children have the same polarity as~$v$. See Figure~\ref{fig:stripes2}.
    If the algorithm finishes the subproblem by completing verification and concludes that $b_v$ and $r_v$ are representing neighboring stripes of the output, then $v$ is a leaf of the tree.
\end{definition}

\begin{figure}[ht]
    \centering
    \includegraphics[scale=0.4]{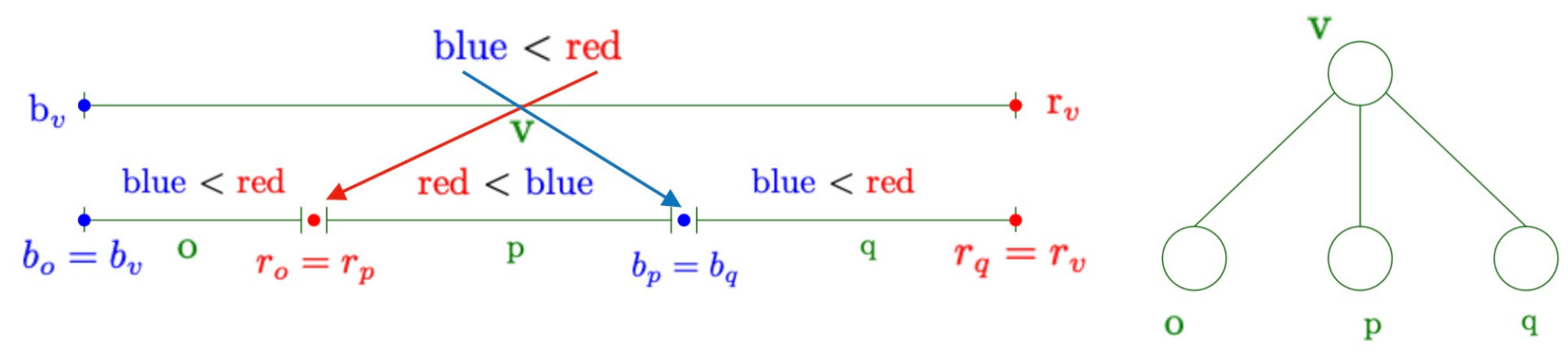}
    \caption{A parent node $v$ and its three children $o,p,q$ in the refinement tree. 
    The pivots for the parent are $(b_v,r_v)$, and since the parent is of type \(\hbox{blue}<\hbox{red}\), the inversion found is of type \(\hbox{red}<\hbox{blue}\). 
    The left and right children have the same polarity as the parent, whereas the middle child gets opposite polarity.}
    \label{fig:stripes2}
\end{figure}
To analyse the refinement tree, it is convenient to work with a totally ordered list~$L_I = (x_0,x_1,\ldots x_{n+m-1})$ of all red and blue elements. 
In the case of bichromatic sorting, this total order is given by the sorted ordering. 

\begin{definition}\label{def:listsubinstance}
    [List subinstance, Stripe subinstance and subproblem]
    Let $r\in R$ and $b\in B$ be two elements and let $i,j$  be their indices in $L_I$, 
    i.e., $x_i=r$ and $x_j=b$.
    Then, the instance $I_{rb}$ given by the list $(x_i,\ldots,x_j)$ or $(x_j,\ldots,x_i)$ is called the \emph{list subinstance} of $L_\I$ defined by $r$ and $b$. 
    Let~$\I$ be an instance of bichromatic sorting and let $x_o,x_p,x_q,x_r$ be four consecutive pivots on the backbone of InversionSort at some point in time.
    Then the \emph{stripe subinstance} of $x_p$ and $x_q$ consists of the stripes of $x_p$ and~$x_q$ and all elements between $x_p$ and $x_q$.
    The \emph{subproblem} consists of the two buckets of $x_p$ and $x_q$, i.e., the elements of the color of $x_p$ between $x_o$ and~$x_q$, and the elements of the other color (that of $x_q$) between $x_p$ and $x_r$.
\end{definition}

Observe that at any snapshot of the algorithm, any existing backbone corresponds to the in-order leaf traversal of some subtree  of the refinement tree that includes the root.\footnote{This subtree is obtained by removing from the full refinement tree subtrees of those internal nodes that have not yet found inversions.} 
This traversal consists of internal nodes and leafs of the complete refinement tree.
See Figure~\ref{fig:tree}.

\begin{figure}[ht]
    \centering
    \includegraphics[width=14cm]{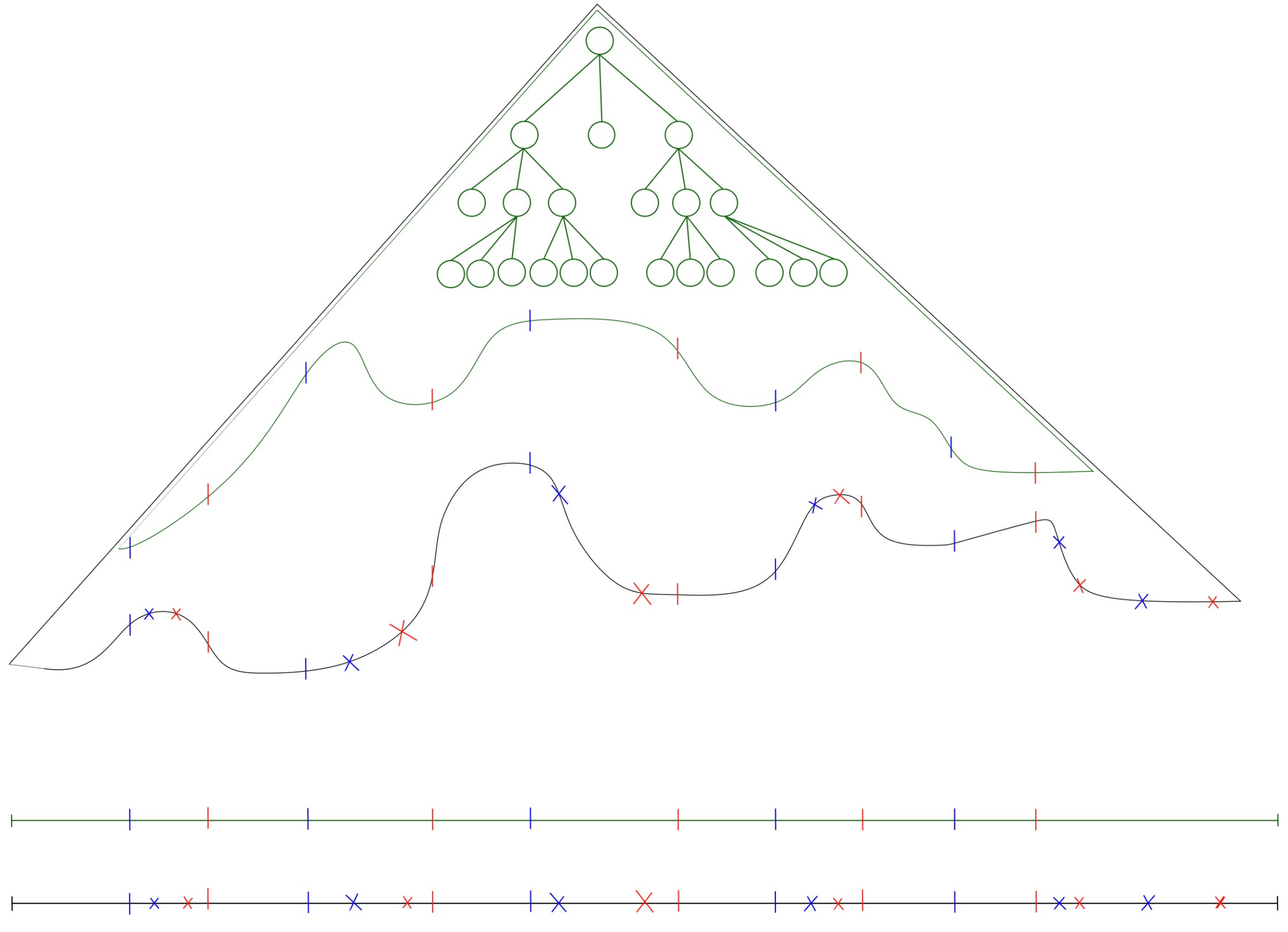}
    \caption{A sketch of the complete ternary refinement tree (in black) and the partial refinement tree corresponding to a snapshot of the backbone during a run of InversionSort.}
    \label{fig:tree}
\end{figure}

\subsection{Inversion Finding}
We start with a central insight about the randomness of inversion finding that has important consequences for the depth of the refinement tree of InversionSort.

\begin{restatable}[Randomness in Inversion Finding]{lemma}{lemInvSortUniform}\label{lem:uniform}
At any stage of the InversionSort, consider a successful inversion finding procedure, which finds an inversion $y<x$ between representatives $u_i<y<x<u_{i+1}$. 
Say, w.l.o.g., that $u_i$ is red and $u_{i+1}$ is blue, and hence $x$ is red and $y$ is blue.
\begin{enumerate}
    \item for any~$y\in X_{i+1}$, 
    conditioned on $y$ being in the inversion, we have $P(x=a) \le P(x=b)$ for $a<b$,
    i.e., $x$ is uniformly distributed among all the red elements in $R_y = \{x \in X_i \mid y<x<u_{i+1}\}$ or  biased towards the outside of $R_{y}$.
    \item for any~$x\in X_{i}$, conditioned on $x$ being in the inversion, $y$ is either uniformly distributed in $B_x=\{y \in X_{i+1} \mid u_i<y<x\}$, or biased towards the outside, i.e., $P(y=a) \ge P(y=b)$ for $a<b$. 
    \item Either statement (1) holds with the uniform distribution, or statement (2) holds with the uniform distribution.
\end{enumerate}
\end{restatable}
\begin{proof}
    Conditioned on one endpoint of the inversion being fixed, say the red $x$.
    Then there is a set~$B_x$ of blue elements that can form an inversion with~$y$ (if not, $x$ cannot be part of an inversion). 
    If the inversion was found by a uniformly sampled pair, each~$y\in B_x$ is chosen with the same probability, showing the statement of the lemma.
    If the inversion was found after sampling some of the elements, we actively choose one of the endpoints uniformly at random.
    \riko{reviewer 2 comment $x$ belongs to $X_i$ - check this}
    If this randomly chosen element is~$x$, we get (1.) with the uniform distribution.
    Hence, for the second part, we should assume that $y$ is the minimum (by the choice of naming in the lemma) of a set~$S_B$ of sampled blue elements.
    Observe that $y \in S_B\cap B_x$, in fact $y = \min (S_B\cap B_x)$.
    Because $S_B$ is the result of uniform sampling, $y$ is biased in the direction expressed in the lemma.
\end{proof}
\subsection{Bounding the depth of the refinement tree}

\begin{definition}
    Let~$v$ be a node in the refinement tree with the pivots $r_v$ and $b_v$, and let $I_{r_vb_v}$ the list-subinstance of~$v$.
    Define $R_v=I_{r_vb_v} \cap R$, $B_v= I_{r_vb_v}\cap B$, and
    $P_v$ as the \emph{number of pairs} of~$v$ as %
    $P_v = |R_v|\cdot|B_v|$.
\end{definition}

\begin{lemma}\label{lem:productProgress}
    Let~$v$ be a non-leaf node of the refinement tree of bichromatic InversionSort and let $o,p,q$ be its children.
    Then, with probability at least 1/2, 
    \[
        \max_{w\in\{o,p,q\}} P_w \le \frac{7}{8} P_v
    \]
\end{lemma}
\begin{proof}
    Let $b\in B$, $r\in R$ be the inversion  and assume w.l.o.g $r_v<b<r<b_v$, 
    and call the three children of~$v$ as $o,p,q$ with 
    $r_o=r_v$, $b_o=b_p=b$, $r_p=r_q=r$ and $b_q=b_v$. See Figure~\ref{fig:stripes2}.
    W.l.o.g., assume $b$ is uniformly chosen given~$r$, but $r$ has a bias towards making the outer child~$q$ smaller, i.e., since $r$ is a max, $|\P_{q}|$ is smaller than if $r$ were uniform.
    By Lemma~\ref{lem:uniform}, we can conclude that with probability at least 3/4,
    $\max(|B_o|,|B_p|)\le 7/8(|B_o|+|B_p|) \le 7/8 |B_v|$, and $|R_q| \le 3/4 |R_v|$ with probability at least 3/4.
    By a union bound, we have both statements with probability at least 1/2,
    and in the following we assume they both hold.
    
    From the first statement follows (even if $R_o=R_v$ or $R_p=R_v$) that $P_o \le 7/8 P_v$ and $P_p\le 7/8 P_v$.
    From the second statement follows $P_q \le 3/4 P_v$ (even if $B_q=B_v$).
\end{proof}

With these tools in our hands, we can go back to
\thmheight*
\begin{proof}
    Let~$x$ be an element of~$\I$ and consider the root-to-leaf path~$P$ of~$\T$ that consists of the nodes~$v$ that contain~$x$ in~$\I_v$.
    Let $Q\subset P$ be the set of nodes where the number of pairs is reduced by a factor 7/8, as in Lemma~\ref{lem:productProgress}.
    Then $|Q|=\Theta(\log N)$ and by a Chernoff bound\riko{rev2: see details}, with probability at least $1-N^{-2}$, $|P| = O(|Q|)$.
    A union bound over all~$x$ gives the statement of the theorem.
\end{proof}

\subsection{Pivoting Cost Analysis}

Recall that InversionSort always pivots with the newly found inversion, and in doing so, it compares all elements in a bucket $X_{i}$ to at most two new pivots. 
Even though InversionSort does not proceed layer-by-layer (in fact, different vertices of the same layer could have wildly varying times of InversionSort finding inversions in them and refining them into children; see Figure~\ref{fig:tree}), we can still upper bound the total pivoting cost performed on a layer by $O(n+m)$. Using Theorem~\ref{thm:height}, we get that

\begin{theorem}[Pivoting Cost of InversionSort]\label{thm:pivot}
    Let $N=n+m$. With probability at least $1-1/N$, the pivoting cost of InversionSort on an instance of size $N$ is at most $O(N \log N)$.
\end{theorem}

\subsection{Inversion Finding Cost Analysis}\label{sec:inversioncost}

What remains to be analysed is the cost, in number of comparisons, that is incurred to find the inversions, i.e., pairs of new pivots, as described in Section~\ref{sec:inversionSearch}

For every stage of InversionSort, the current backbone corresponds to a left-to-right traversal of the current subtree (including the the root) of the refinement tree.

Recalling Definition~\ref{def:listsubinstance}, such an inversion search is performed for each subproblem, we search for an inversion in the list subinstance, and, at least initially or if monochromatic comparisons are expensive or disallowed, the stripe subinstance is the structure the algorithm is faced with.
Note that elements of the subproblem that are not in the subinstance make inversion finding more difficult because such elements cannot form an inversion, but the algorithm cannot distinguish them from elements in the subinstance.
We say that such elements in the bucket of one pivot \emph{spill in} to the subproblem.
This spill-in is the shaded part in Figure~\ref{fig:stripes1}.
\begin{definition}
    [unaffected subproblem]\label{def:unaffected}
    Let $r,b$ be two neighboring pivots on the backbone in some round of inversion search of InversionSort.
    Then the subproblem $r,b$ is \emph{unaffected} if the number of red elements in the stripe subinstance of $r,b$ is at least 1/4 of the number of red elements in the subproblem of $r,b$ and the number of blue elements in the stripe subinstance of $r,b$ is at least 1/4 the number of blue elements in the subproblem of $r,b$.
\end{definition}

Note that in the course of running InversionSort, by refining the backbone, subproblems can turn from affected to unaffected, but not vice versa.

Observe that neighboring subinstances share the stripe of the pivot; this stripe does not constitute spill-in.
Still, over a single pivot, there typically is spill-in of the same color in both directions, but only one of them can affect the receiving subproblem:
$R_k$ consists of the red elements in the two neighboring subproblems, $S_i\subset R_k\cup\{r_k\}$ \riko{rev 2 i is k? no, I don't think so} is part of both, so one of the two sets of spilling-in elements must be less than half of~$R_k$, so clearly not both of them can be more than 3/4 of $R_k$.

\subsection{Constant local competitiveness of inversion finding}
\label{sec:localCompetitive}
In this section, we will argue that inversion finding for unaffected subinstances is justified by the cost of the Hamiltonian in the subinstance.
In Section~\ref{sec:inversionSearch}, we already argued that non-successful inversion searches that lead to a certificate that there are no further inversions are justified by parts of the Hamiltonian. 
Here, we (also) consider successful searches, and relate their cost to the cost of the Hamiltonian between and including the stripes of the representatives InversionSort has as neighbors on the backbone.

The proofs of the following Lemmas and Theorems are given in Appendix~\ref{apx:om}.

\begin{restatable}{lemma}{lemlocalcomp}
    [local competitiveness of inversion finding]\label{lem:localComp}
    Let $r,b$ be two neighboring pivots of InversionSort, and let $t$ be the number of rounds (up to a constant the cost of comparisons) until InversionSort finds an inversion or concludes that there is none, starting from the first round when the subproblem of $r,b$ is unaffected.
    Let $H_{r,b}$ be the cost of the Hamiltonian of the stripe subinstance between the stripe of $r$ and the stripe of $b$.
    Then $\mathbb{E}[t]=O(H_{r,b})$, and $\text{Pr}[t \leq \mathbb{E}[t]+t'] \geq 1-(1/8)^{t'}.$
\end{restatable}

\subsubsection*{Overall cost of inversion finding in unaffected subproblems}

Observe that each layer in the refinement tree almost partitions the original instance into subinstances: Only the buckets of the pivots on the backbone are counted twice.
In other words, if we take all subinstances of the same polarity (every other) of one level of the refinement tree, they are disjoint, and the sum of the cost of the Hamiltonian in the subinstances is at most the cost of the Hamiltonian. 
By Theorem~\ref{thm:height} we can hence bound the total cost of inversion search / finding in unaffected subproblems to be only a $O(\log N)$ factor higher than the cost of the Hamiltonian.

\begin{restatable}{theorem}{thmtotalUnaffected}\label{thm:totalUnaffected}
    Let $u$ be the total cost of comparisons that InversionSort performed in unaffected subproblems when running on instance~$\I$ of size~$N$, and let $H_{\I}$ be the cost of the Hamiltonian of~$\I$.
    Then, with probability at least $1-1/N$, $u=O(H_{\I} \log N)$.
\end{restatable}

\subsection{Many unaffected subproblems}\label{sec:manyUnaffected}

What remains to be shown is that there are sufficiently many unaffected subproblems, such that the progress made there, by the previous section, justifies the cost of all performed comparisons.
The setup is that we consider some current state of InversionSort as defined by one round of inversion search in all subproblems. 
We will show that for each unaffected subproblem there are at most $O((\log N)^2)$ affected ones (creating the same cost, but not necessarily making progress).

A middle child with its three children currently being leafs is easy to analyse because the inversion that defines the boundaries between the children is described by Lemma~\ref{lem:uniform}.
We will use the following lemma on partial refinement trees that describe any backbone InversionSort might work with at some time.
\begin{restatable}{lemma}{lemmiddlechild}
    \label{lem:middleChild}
    Let~$T$ be a (subtree of a) ternary tree of height~$h$ and let $L=(v_1,\ldots,v_k)$ be part of the left-to-right leaf-traversal of~$T$ of length~$k\ge 4(h+2)$.
    Then $L$ must contain a middle child whose 3 children are leaves.
\end{restatable}

Observe that finished pairs of neighboring pivots on the backbone, i.e., subproblems where InversionSort verified that there are no further inversions, don't spill out to their neighbors, and remember that over one pivot,  spilling can only be in one direction.

\begin{restatable}{lemma}{lemoneUnaffected}
    \label{lem:oneUnaffected}
    Let~$L$ be a list of consecutive subproblems on the backbone, and assume there is no spill-in from the left to the first or from the right to the last (because of finished subproblems or the instance ending).
    Then at least one of the subproblems must be unaffected.
\end{restatable}

\begin{restatable}{lemma}{lemnoLongWO}[Not too many affected subproblems]
    \label{lem:noLongWO}
    Let~$a$ be the number of active subproblems on the backbone at some stage of InversionSort and let~$u$ be the number of the unaffected subproblems.
    There exists a constant~$c$  
    such that $\text{Pr}[a \le c u  (\log N)^2] \geq 1-1/N^3$.
\end{restatable}

\subsection{Putting everything together: Proof of Theorem~\ref{thm:bichromaticinversionsort}}

With these ingredients, we are ready for a proof of

\thmbichromaticinversionsort*
\begin{proof}
    The cost of InversionSort consists of pivoting and inversion finding.
    Theorem~\ref{thm:pivot} shows that pivoting uses $O(n\log n)$ bichromatic comparisons.
    Observe $H_\I\ge n-1$, hence pivoting costs $O(H_\I\log N)$. 
    Let $v$ be the cost of inversion finding, and $u$ be the cost of such comparisons in unaffected subproblems.
    Lemma~\ref{lem:noLongWO} shows $v \le u (\log N)^2$ with probability $1-1/N$ (union bound over the at most $N^2$ many rounds) and by Theorem~\ref{thm:totalUnaffected} follows $u=O(H_{\I} \log N)$, showing the statement of the theorem.
\end{proof}
 
\newpage
\bibliographystyle{plainurl}

\begin{thebibliography}{10}

\bibitem{alon1994matching}
Noga Alon, Manuel Blum, Amos Fiat, Sampath Kannan, Moni Naor, and Rafail
  Ostrovsky.
\newblock Matching nuts and bolts.
\newblock In {\em Proceedings of the fifteenth annual ACM-SIAM Symposium on
  Discrete Algorithms (SODA'94)}, pages 690--696, 1994.

\bibitem{banerjee_et_al:LIPIcs:2016:6044}
Indranil Banerjee and Dana Richards.
\newblock {Sorting Under Forbidden Comparisons}.
\newblock In 15th Scandinavian~Symposium and Workshops on~Algorithm Theory
  (SWAT~2016), editors, {\em Rasmus Pagh}, volume~53 of {\em Leibniz
  International Proceedings in Informatics (LIPIcs)}, pages 22:1--22:13,
  Dagstuhl, Germany, 2016. Schloss Dagstuhl--Leibniz-Zentrum fuer Informatik.
\newblock URL: \url{http://drops.dagstuhl.de/opus/volltexte/2016/6044}, \href
  {https://doi.org/10.4230/LIPIcs.SWAT.2016.22}
  {\path{doi:10.4230/LIPIcs.SWAT.2016.22}}.

\bibitem{bender2021batched}
Michael~A Bender, Mayank Goswami, Dzejla Medjedovic, Pablo Montes, and Kostas
  Tsichlas.
\newblock Batched predecessor and sorting with size-priced information in
  external memory.
\newblock In {\em Latin American Symposium on Theoretical Informatics}, pages
  155--167. Springer, 2021.

\bibitem{charikar2002query}
Moses Charikar, Ronald Fagin, Venkatesan Guruswami, Jon Kleinberg, Prabhakar
  Raghavan, and Amit Sahai.
\newblock Query strategies for priced information.
\newblock {\em Journal of Computer and System Sciences}, 64(4):785--819, 2002.

\bibitem{deshpande2014approximation}
Amol Deshpande, Lisa Hellerstein, and Devorah Kletenik.
\newblock Approximation algorithms for stochastic boolean function evaluation
  and stochastic submodular set cover.
\newblock In {\em Proceedings of the twenty-fifth annual ACM-SIAM Symposium on
  Discrete Algorithms (SODA'14)}, pages 1453--1466. SIAM, 2014.

\bibitem{goswami2023sorting}
Mayank Goswami and Riko Jacob.
\newblock Sorting with priced comparisons: The general, the bichromatic, and
  the universal, 2023.
\newblock \href {http://arxiv.org/abs/2211.04601} {\path{arXiv:2211.04601}}.

\bibitem{gupta2001sorting}
Anupam Gupta and Amit Kumar.
\newblock Sorting and selection with structured costs.
\newblock In {\em Proceedings 42nd IEEE Symposium on Foundations of Computer
  Science (FOCS'01)}, pages 416--425. IEEE, 2001.

\bibitem{6108244}
Zhiyi Huang, Sampath Kannan, and Sanjeev Khanna.
\newblock Algorithms for the generalized sorting problem.
\newblock In {\em 2011 IEEE 52nd Annual Symposium on Foundations of Computer
  Science (FOCS'11)}, pages 738--747, 2011.
\newblock \href {https://doi.org/10.1109/FOCS.2011.54}
  {\path{doi:10.1109/FOCS.2011.54}}.

\bibitem{kaplan2005learning}
Haim Kaplan, Eyal Kushilevitz, and Yishay Mansour.
\newblock Learning with attribute costs.
\newblock In {\em Proceedings of the thirty-seventh annual ACM symposium on
  Theory of computing (STOC'05)}, pages 356--365, 2005.

\bibitem{komlos1998matching}
J{\'a}nos Koml{\'o}s, Yuan Ma, and Endre Szemer{\'e}di.
\newblock Matching nuts and bolts in {O}$(n \log n)$ time.
\newblock {\em SIAM Journal on Discrete Mathematics}, 11(3):347--372, 1998.

\bibitem{kuszmaulnarayanan}
William Kuszmaul and Shyam Narayanan.
\newblock Stochastic and worst-case generalized sorting revisited.
\newblock In {\em 2021 IEEE 62nd Annual Symposium on Foundations of Computer
  Science (FOCS'22)}, pages 1056--1067. IEEE, 2022.

\bibitem{mozes2008finding}
Shay Mozes, Krzysztof Onak, and Oren Weimann.
\newblock Finding an optimal tree searching strategy in linear time.
\newblock In {\em In Proceedings of the Nineteenth Annual ACM-SIAM Symposium on
  Discrete Algorithms (SODA'08)}, volume~8, pages 1096--1105, 2008.

\bibitem{onak2006generalization}
Krzysztof Onak and Pawel Parys.
\newblock Generalization of binary search: Searching in trees and forest-like
  partial orders.
\newblock In {\em 2006 47th Annual IEEE Symposium on Foundations of Computer
  Science (FOCS'06)}, pages 379--388. IEEE, 2006.

\bibitem{purohit2018improving}
Manish Purohit, Zoya Svitkina, and Ravi Kumar.
\newblock Improving online algorithms via ml predictions.
\newblock {\em Advances in Neural Information Processing Systems}, 31, 2018.

\bibitem{rawlins1992compared}
Gregory~JE Rawlins.
\newblock {\em Compared to what? An introduction to the analysis of
  algorithms}.
\newblock Computer Science Press, Inc., 1992.

\end{thebibliography}

\newpage

\appendix

\noindent\textbf{\Huge{Appendices}}

\section{Variations of bichromatic costs}\label{sec:other_bichromatic}
\subsection{Bichromatic most expensive: $\alpha<1$ and $\beta<1$}

In this case, it is natural to first sort both colors with monochromatic comparisons, to then do binary merging using exponential searches.
It is easy to see, that this algorithm is within an~$O(\log N)$ factor of the cost of the Hamiltonian.

\subsection{Bichromatic middle expensive: $\alpha<1<\beta$ or $\beta<1<\alpha$}

This setting can be expressed as a monotone function in the framework of~\cite{gupta2001sorting}.

There is also a straight forward direct  algorithm:
W.l.o.g., assume $\alpha<1<\beta$.
Sort the red elements using  red-red comparisons at cost~$O(\alpha |R|\log |R|)$.
For each blue element, perform a binary search in these red, at total cost of~$O(|B|\log |R|) $.
Finally, sort the blue stripes with total cost $O(\beta \sum |B_i|\log |B_i|)$.
Each of the three terms is justified by the need to determine the rank of a certain type of elements, and the above procedure uses the cheapest available comparisons for this task.

\section{More than two colors}\label{sec:multipartite}
\begin{definition}
    [Multichromatic Sorting]\label{def:multichromatic}
    Multichromatic Sorting is universal sorting, where the elements are colored with~$k$ different colors. All bichromatic comparisons have cost~1, and for each color~$i$, the monochromatic comparisons have a cost~$\gamma_i>1$, where $\gamma_i=\infty$ means that these comparisons are forbidden and stripes of this color are reported as a whole.
\end{definition}

To solve the multichromatic setting, 
run a natural variant of (bichromatic) InversionSort:
Pick one random elements from each color, sort these, and take this as the backbone.
Consider the remaining elements in random order.
Place the element into the backbone using binary search. 
If this leads to a new element on the backbone (neighbors are of different colors), pivot the two current buckets.
This maintains the invariants of the backbone and that three consecutive elements on the backbone have three different colors.
Hence, all neighboring buckets are starting points of bi-chromatic problems, and we can use InversionSort on them.
Here we interleave the computation as if the whole backbone was one backbone in the bichromatic setting.
 
\section{Omitted Proofs}
\label{apx:om}

\lemlocalcomp*
\begin{proof}
    Let~$C$ be the number of rounds of comparisons for which subproblem~$r,b$ is active.
    Observe that InversionSort does on this subproblem $C$ bichromatic comparisons, and samples $\sim C/\alpha$ red elements and finds their max and min in $O(C)$ cost. 
    Similarly, the blue sample consists of $\sim C/\beta$ elements, incurring cost $O(C)$ for blue-blue comparisons.
    This is from the round in which $r,b$ are both on the backbone until either an inversion between~$r,b$ is found, or a certificate that there is no such inversion is established with comparisons.
    Hence, $C$ is asymptotically upper bounded by the cost of the cheapest proof, if there were no further inversions.
    We will show that these comparisons are justified by edges on the Hamiltonian, by case analysis.
    Assume the Hamiltonian on the subinstance between~$r$ and $b$ consists of $h_r$ red and $h_b$ blue edges.
    Let $n_r\ge 4$ and $n_b\ge 4$ the number of red and blue elements in the subinstance.
    Otherwise, if one of the colors has at most~4 elements, doing all bichromatic comparisons is $O(n_r+n_b) = O(\OPT)$.
    W.l.o.g. (for naming), assume that the red pivot is to the left and the blue is to the right, i.e., we are searching for an inversion blue is left of red.
    
    Observe that if there are at least $n_r/4\ge 1$ red-red edges on the Hamiltonian in the subinstance, i.e. $h_r\ge n_r/4$, the cost of finding the rightmost red element is justified, i.e., at the latest after doing comparisons for cost $\Theta( \alpha n_r) = O(\OPT)$, InversionSort has identified this element.
    Otherwise, there are at most $h_r < n_r/4$ red-red edges on the Hamiltonian. 
    If there were no red-red edges, there would be~$n_r$ red stripes, each red-red edge reduces this number by one, so there are more than $n_r-n_r/4=3/4 n_r$ (red) stripes, and a total of at least $3/4 n_r$ elements (because each stripe has at least one element) in stripes different from the leftmost stripe (possible inversions). 
    In this case, the right half of the (red) stripes must contain at least $3/8 n_r$ red elements.
    Similarly, if there are $n_b/4$ blue-blue edges on the Hamiltonian, finding the leftmost blue is justified, or the left half of the stripes contains at least $3/8 n_b$ blue elements. 
    
    Now consider the different cases.
    Assume $h_r \ge n_r/4$ red-red and $h_b \ge n_b/4$ blue-blue edge on the Hamiltonian in the subinstance.
    Then finding the rightmost red and the leftmost blue is justified, and InversionSort finds these two elements and hence an inversion (or a proof that there is none) in cost at most $3(\alpha n_a+\beta n_b) =O(\OPT)$. 
    
    Assume $h_r \ge n_r/4$ red-red and $h_b < n_b/4$ blue-blue edges on the Hamiltonian.
    Then, after comparisons for cost $\Theta(\alpha n_r) = O(\OPT)$, InversionSort has identified the rightmost red element~$r_m$, and there must be at least $3/4 b$ blue elements left of~$r_m$.
    From then on, the probability of InversionSort finding an inversion in one round is at least $3/4 > 1/8$.
    The symmetric argument holds if the role of red and blue is interchanged.
    
    The last case is $h_r < n_r/4$ red-red and $h_b < n_b/4$ blue-blue edges on the Hamiltonian.
    Then the probability of finding an inversion that crosses the middle is $(3/8)^2$, and InversionSort finds an inversion with expected $O(1)$ cost ($O(\log N)$ whp\riko{rev 2 what does it mean}) (relying only on bichromatic inversion search).

    Hence, at the latest after a setup of cost $H_{r,b}$, the probability of finding an inversion is at least $(3/8)^2 = 9/8^2 > 1/8 $.
\end{proof}

\thmtotalUnaffected*
\begin{proof}
    The refinement tree is a tree with at most~$N$ leaves and hence at most $N-1$ internal nodes, and the additive terms from these Bernoulli experiments total to $O(N\log N)$.
    As argued in the preceding paragraph, the setup cost of each of the $O(\log N)$ layers of the refinement tree are $O(H_{\I})$ each.
\end{proof}

\lemmiddlechild*
\begin{proof}
    If there are 5 consecutive vertices in~$L$ of the same depths, 3 of them must be the children of the same vertex, showing the lemma.
    Otherwise, no such 5 consecutive vertices exist.
    Then, by the bound on~$k$, there must be at most 4 vertices (internal, i.e. not starting with~$v_1$) that are a local maximum with depths~$D$, i.e., the preceding and following element in~$L$ have smaller depths.
    The first vertex with depths~$D$ must be the left child of some node~$v$, and it is a leaf.
    Hence, its two right siblings exist, and because the next vertices on~$L$ have depths~$\le D$, the siblings cannot have children. 
    This shows the statement of the lemma.
\end{proof}

\lemoneUnaffected*
\begin{proof}
    There are $|L|-1$ internal pivots over which a spill-in into subproblems could happen, but there are $|L|$ subproblems.
    By the pigeon-hole principle, there must be a subproblem without spill-in, i.e., unaffected.
\end{proof}

\lemnoLongWO*
\begin{proof}
    Let~$L$ be a list of consecutive active subproblems on the backbone. 
    If we can show the statement of the lemma for each individual such list, the lemma follows.
    
    If $a=|L|\le c (\log N)^2$, the claim follows from Lemma~\ref{lem:oneUnaffected} because there is at least one unaffected subproblem.
    
    Otherwise let $h=O(\log N)$ be the height of the refinement tree.
    By Lemma~\ref{lem:middleChild}, there are at least $a'=a/(4h+2)$ many subproblems corresponding to a vertex whose children are leafs.
    Now, by Lemma~\ref{lem:uniform}, one of the boundaries is chosen uniformly at random, and its probability of having spill-over to the left is 1/4, and so is for spill-over to the right.
    Now encode, from left to right, these spill-overs by a 1-bit if it is different from the previous.
    Let~$k$ be the number of 1s in this sequence, and observe that the number of ones is a lower bound on the number of subsequences with the same direction of spill-over. 
    Every other of these subsequences must have an unaffected subproblem.
    $\EXP[k]=a'/2 \geq c (\log N)^2) /h \ge 2\log N$ if the constant~$c$ is sufficiently large.
    Hence, by a Chernoff bound, with probability~$1-N^3$, \riko{rev 2: show details}
    we get $u\ge a'/4 \ge a/(12h) \gg a/c (\log N)^2$
    and $u\ge a/c (\log N)^2$ for sufficiently large~$c$.
\end{proof}

\end{document}